\documentclass[a4paper,11pt]{article}
\usepackage[utf8]{inputenc}

\usepackage{cite}
\usepackage{graphicx}
\usepackage{graphicx}
\usepackage{pstricks}
\usepackage{appendix}
\usepackage{dsfont}
\usepackage{amsmath}
\usepackage{amssymb}
\usepackage{amsthm}
\usepackage{mathtools}
\usepackage{comment}
\usepackage[shortlabels]{enumitem}
\usepackage{tikz,lipsum,lmodern}
\usepackage[most]{tcolorbox}
\usepackage{hyperref}
\usepackage{algorithm}
\usepackage{algpseudocode}
\usepackage{algpseudocode}
\usepackage{xcolor}
\usepackage{blindtext}

\numberwithin{equation}{section}
\newtheorem{theorem}{Theorem}[section]

\newtheorem{proposition}[theorem]{Proposition}
\newtheorem{lemma}[theorem]{Lemma}

\theoremstyle{remark}
\newtheorem{remark}[theorem]{Remark}

\DeclarePairedDelimiterX \ip[2]{\langle}{\rangle}{#1,#2}
\DeclarePairedDelimiterXPP \Prob[1]{\mathbb{P}}\{\}{}{\newcommand\given{\nonscript\:\delimsize\vert\nonscript\:\mathopen{}} #1} 
\DeclarePairedDelimiterXPP \Probevent[1]{\mathbb{P}}(){}{\newcommand\given{\nonscript\:\delimsize\vert\nonscript\:\mathopen{}}#1} 

\DeclareMathOperator{\softmax}{softmax}

\newcommand{\E}{{\mathbb E}}

\newcommand{\e}{\varepsilon}

\newcommand{\one}{{\mathds{1}}}

\title{LLM Watermarking Using Mixtures and Statistical-to-Computational Gaps}

\begin{document}
\author{
  Pedro Abdalla
  \qquad
  Roman Vershynin\thanks{Department of Mathematics, UC Irvine}
}

\maketitle

\begin{abstract}
Given a text, can we determine whether it was generated by a large language model (LLM) or by a human? A widely studied approach to this problem is watermarking. We propose an undetectable and elementary watermarking scheme in the closed setting. Also, in the harder open setting, where the adversary has access to most of the model, we propose an unremovable watermarking scheme. 
\end{abstract}

\section{Introduction}
Large Language Models (LLMs) have emerged as a powerful technology for generating human-like text \cite{achiam2023gpt,stokel2022ai}. On one side, an LLM performs well if it produces text that closely resembles human writing. On the other side, the use of high-performance LLMs also bring undesirable consequences such as the spread of misinformation \cite{starbird2019disinformation}, misuse in education \cite{stokel2022ai,milano2023large}, and data pollution \cite{radford2023robust,shumailov2023curse}.

In this context, there is an urge to develop methods to distinguish human and AI generated text to mitigate those outcomes. One prominent technique is the so-called watermarking approach in which the goal is to embed a detectable signal in the text generated by the LLM. 

Before describing watermarking in more details, we recall the concept of tokenization. In a nutshell, a word consists in small pieces of ``sub-words" known as tokens. A LLM outputs each token sequentially by computing a probability distribution over a fixed set of possible tokens (dictionary) and sampling the next token from it. The distribution of the next token varies from token to token as it depends on the previous tokens sampled, while the dictionary remains fixed during the process of text generation.

The most common approach to watermark a text is to watermark each token by planting a hidden structure into its probability distribution. In this sense, it is natural to impose some requirements on which properties a good watermarking should have. For example, it is natural to require that the watermarking scheme does not deteriorate the quality of the text or that it cannot be easily removed by someone with malicious intentions, an adversary.  

In what follows, we describe the requirements for our watermarking scheme. To this end, we shall make a distinction between two different settings: The closed setting and the open source setting. We first describe the closed setting. In this case, one would like to generate a watermarking satisfying three requirements
\begin{itemize}
    \item {\bf Statistical Undetectability:} Any  algorithm based solely on the text generated by the LLM fails to detect any change in the probability distributions used to generate the tokens \footnote{It should not be confused with the notion of cryptographic undetectability as in \cite{christ2024undetectable}.}.
    \item {\bf Completeness:} It is possible to detect the watermarked model if the algorithm has access to extra piece of information known as ``secret key".
    \item {\bf Soundness:} Any text generated independently of the secret key has negligible chance of being detected as watermarked.
\end{itemize}
We postpone the mathematical framework of those requirements to the next section. Now, let us provide some intuition behind those requirements. The first requirement is useful to preserve the quality of the text generated and prevents malicious users (adversaries) to manipulate the text to remove the watermarking scheme. The second requirement is the core idea of watermark to distinguish texts generated by AI and humans which clearly requires a ``secret key", otherwise would contradict the ``undetectability" requirement. Finally, the last requirement is of fundamental importance as, for example, it prevents false accusations of AI misuse (see for example \cite{fowler2023we}).

A harder task is to watermark the text in the so-called open source setting. This is motivated by the recent explosion of AI open source models, where the user has access to the model parameters and the associated code \cite{touvron2023open,zhang2022opt,christ2024provably}. Since now the adversary has much more power, we replace the ``undetectability" requirement by the weaker ``unremovability" requirement:
\begin{itemize}
    \item {\bf Unremovability:} Any adversary that does not have knowledge about the secret key cannot remove the watermarking scheme unless it deteriorates the quality of the text.
\end{itemize}
Clearly, one has to impose some conditions on what the adversary knows, otherwise he could train a new model on its own, making watermarking impossible. Besides, the adversary goal is to remove the watermark and use the text for malicious purposes, so the quality of the text cannot be deteriorated too much.

In this work, we allow the adversary to arbitrarily modify the inputs of the LLM and also allow him to have knowledge of each token distribution used for sampling (after the watermarking scheme was planted).

\subsection{Related Work}
Several watermarking schemes were proposed \cite{abdelnabi2021adversarial,munyer2023deeptextmark,yoo2023robust} for the closed setting without any formal guarantee. Perhaps, the first watermarking scheme with provable guarantees is from \cite{kirchenbauer2023watermark}, where the authors proposed to split the dictionary into a green list and red list. The probabilities corresponding to the tokens in the green list are slightly increased while the ones in the red list are slightly decrease. The watermarking can be detected by checking the frequency of tokens in the green list versus tokens in the red list. Therefore, the downside of this approach is that the ``undetectability" requirement is not fulfilled.

Another line of work \cite{christ2024undetectable,li2025statistical,aaronson2023watermarking,kuditipudi2023robust} is dedicated to the following idea: Let $u$ be a random variable distributed uniformly over the interval $[0,1]$ and $p=(p(1),\ldots,p(d))$ be a probability distribution over a dictionary of size $d$. We can sample the next token according to $p$ by sampling from $u$ first and then observing that for any $k\in \{1,\ldots,d\}$
\begin{equation*}
    \mathbb{P}\Bigg\{u\in \Bigg[\sum_{i=1}^{k-1}p(i),\sum_{i=1}^{k}p(i)\Bigg]\Bigg\} = p(k).
\end{equation*}
The watermarking schemes exploit correlation between the tokens and the corresponding (uniform) random variables $u$'s used to sample the tokens. Despite this approach has some guarantees, the major drawback is that the detection algorithm is quite convoluted relying on a complicated optimization because it is hard to capture the planted correlation. In addition to this, to achieve the ``undetectability" requirement for the whole text, the approaches in the literature rely on some cryptographic assumptions.

To the best of our knowledge, the only result for the open source setting is from \cite{christ2024provably}. The authors proposed to perturb each logit in the softmax rule (see equation \ref{eq:softmax_rule}) by a vector sampled from a multivariate Gaussian distribution (see \cite{block2025gaussmark} for a similar idea) and exploits the correlation between the text and such vector. The authors provided some theoretical evidence (partially rigorous) for completeness and unremovability under strong assumptions on the text. \footnote{For example, \cite{christ2024provably} assumes that the token distribution behaves as an uniform distribution over a certain subset and also that the adversary makes changes respecting some normalization of the soft-max function which are hard to verify.}

\subsection{Main Contributions}
Our main contributions are new connections between watermarking and robust statistics to derive more efficient watermarking schemes. Our first main result follows from Proposition \ref{prop:undetectability_closed}, Proposition \ref{prop:soundness_closed} and Theorem \ref{thm:completness_closed}, where we propose an elementary watermarking scheme for the closed setting satisfying all the requirements (undetectability, completeness and soundness) under mild assumptions on the distribution of the text. We also argue that the assumptions are necessary in some sense. 

In a nutshell, our watermarking scheme proceeds as follows: In the first step it randomly constructs partitions of the dictionary set into green and red tokens that change at each time a new token is sampled. Similarly to \cite{kirchenbauer2023watermark}, the probabilities of the green tokens are shifted upwards, and the ones for the red tokens, downwards. However, we make some key changes to this shifting scheme, which allows it to achieve undetectability and obtain explicit non-asymptotic guarantees that are easier to interpret than the more involved guarantees available in the literature.

Our second main result follows from Proposition \ref{prop:unremovability}, Proposition \ref{prop:soundness_open} and Theorem \ref{thm:completeness_open} and lies in the realm of the open source setting. By leveraging novel connections to the theory of statistical-to-computational gaps in robust statistics, we proposed a watermarking scheme that is both ``sound" and ``complete", along with mathematically rigorous guarantees. This watermark is also ``unremovable'': any algorithm that attempts to remove it must (indirectly) solve a computationally hard problem -- the sparse mean estimation under Huber's contamination model (see \cite{huber2011robust} for a comprehensive introduction). 

This version of our algorithm also perturbs the logits by Gaussian vectors, similarly to \cite{christ2024provably}. However, instead of using the same perturbation for each token, we use independent perturbations.

Our Gaussian random perturbations are drawn from a mixture of non-centered Gaussians -- a distribution that is hard to distinguish from a centered Gaussian. Thus, the adversary who attempts to remove it faces impossibility results from the robust statistics literature  borrowed from Brennan and Bresler \cite{brennan2020reducibility}. This is a novel approach to study unremovability in watermarking problems.

\subsection{Roadmap}
The rest of this manuscript is organized as follows. In Section \ref{sec:hypo_formulation}, we formally state the watermarking problem in the framework of hypothesis testing. Section \ref{sec:standard} is dedicated to the main results for watermarking in the closed setting  and Section \ref{sec:open} is dedicated to the main result for watermarking in the open source setting. The appendix is dedicated to the proof of technical results.

\section{The Hypothesis Testing Formulation}
\label{sec:hypo_formulation}
Let $T$ be our dictionary of $d$ tokens, which we can identify with $[d] \coloneqq \{1,\ldots,d\}$ without loss of generality.
A text is a sequence of random variables $x_i$ taking values in $[d]$. At each step, the LLM computes logits $L = (\ell(1),\ldots,\ell(d))$ and samples the next token $x \in [d]$ according to the softmax rule
\begin{equation}
\label{eq:softmax_rule}
    p(i)\coloneqq\Prob{x=i}=\softmax(\ell(1),\ldots,\ell(d))\coloneqq\frac{e^{\ell(i)}}{e^{\ell(1)}+\cdots+e^{\ell(d)}}.
\end{equation}

A watermarking scheme consists of two parts: 
\begin{itemize}
    \item The {\em sampling algorithm} at each step tweaks the LLM's output probabilities from $p = (p(1),\ldots,p(d))$ to a new (watermarked) distribution  $q=(q(1),\ldots,q(d))$, using a ``secret key''.

    \item The {\em detection algorithm} takes the whole text $x_1,\ldots,x_N$ and the same secret key, and outputs true/false depending on whether the text was watermarked.     
\end{itemize} 

The sampling algorithm handles undetectability by making sure that $q$ looks like $p$. The detection algorithm handles soundness and completeness by solving the following hypothesis testing problem:

\medskip

\textbf{Detection Algorithm Hypothesis Testing:}
\begin{itemize}
    \item $H_0:$ The text $x_1,\ldots,x_N$ is independent from the watermarking scheme.
    \item $H_a:$ The distribution of the text $x_1,\ldots,x_N$ was sampled from the watermarked distribution.
\end{itemize}

The core challenge in the design of a watermarking scheme is that we need to balance the trade-off between undetectability and completeness without using any prior knowledge about the text. If we simply do not change the LLM's distribution, then the scheme is undetectable but not complete. Similarly, if we make some obvious change to the LLM's distribution, the scheme is complete but detectable. 

\begin{remark}[Closed vs. Open Source Setting]
We just covered the closed watermarking setting. In the open setting, the definition of watermarking as well as the ``completeness'' and ``soundness'' requirements stay the same. But the adversary now has more information than just the text. So instead of ``undetectability'', we require ``unremovability''.

We allow the adversary to modify anything they want about the LLM -- the input, the parameters, even its its architecture. They  can also modify the distributions $q_{x_1},\ldots,q_{x_N}$ used to generate the watermarked text $x_1,\ldots,x_N$ -- but the adversary is not given access to the secret key. 

The goal of the adversary is to approximate the distributions $p_{x_1},\ldots,p_{x_N}$. The ``unremovability" requirement prevents the adversary from achieving this in polynomial time.
\end{remark}

\section{Closed Setting}
\label{sec:standard}
Now, let's formally describe the watermarking scheme in the closed setting. It works for any discrete process on a finite state space, but we will stick to the LLM setting to keep things concrete.

\subsection{The watermarking scheme} \label{s: watermarking scheme closed}

Denote the dictionary (the set of all possible tokens) by $[d] = \{1,\ldots,d\}$ and assume that $d$ is even -- otherwise, add a spurious token. To generate a text consisting of $N$ tokens, do the following for each step $j=1,\ldots,N$.

Randomly and independently paint the dictionary into two colors $1$ and $-1$ with exactly $d/2$ tokens receiving each color. Such a coloring is encoded by the random vector $\Delta_j \in \{-1,1\}^d$. 
Also, generate an independent Rademacher random variable $r_j$, whose purpose is to flip a coloring. The secret key consists of the colorings $r_1\Delta_1,\ldots,r_N\Delta_N$.

Now we will describe the watermark at each given step $j=1,\ldots,N$, and to ease the notation we will drop the subscript $j$.
 
The LLM computes the probability distribution of the next token: 
$$
p = (p(1),\ldots,p(d)).
$$
We will now watermark this distribution by increasing the probabilities $p(i)$ for the tokens $i$ where $r\Delta(i)=1$ and decreasing them where $r\Delta(i)=-1$.

To do that, let $I$ and $I^c$ denote the subsets of tokens colored $1$ and $-1$, respectively, i.e. 
$$
I \coloneqq \{ i \in [d]: \; \Delta(i)=1 \}.
$$
Now we match the tokens of opposite colors according to their rank: the most likely token colored $1$ is matched with the most likely token colored $-1$, the second most likely with the second most likely, and so on. Then, if $r=1$, we increase the probability of the first element of each pair and decrease the probability of the second element. If $r=-1$, we do the opposite. 

Formally, let $w(i)$ denote the $i$-th largest probability of a token with color $1$, i.e. the largest element of the set $\{p(i)\}_{i\in I}$. Similarly, let $\overline{w}(i)$ be the $i$-th largest element of the set $\{p(i)\}_{i\in I^c}$. Since each token $i \in [d]$ is ranked, there exists $i' \in [d/2]$ such that $w(i')=p(i)$ if $i\in I$ or $\overline{w}(i')=p(i)$, otherwise. Set $\e(i)$ to be the smaller probability in each pair of tokens:
\begin{equation} \label{eq: epsilons}
    \varepsilon(i)\coloneqq \min \big( w(i'),\overline{w}(i') \big).
\end{equation}
Next, we compute the watermarked distribution by setting for every $i\in [d]$:
\begin{equation}
\label{eq:vector_q_standard}
    q(i) = p(i) + \varepsilon(i)r \Delta(i).
\end{equation}
Sample token $x \in [d]$ from the probability distribution $q=(q(1),\ldots,q(d))$. 

Repeat this procedure for steps $j=1,\ldots,N$ to generate a text string $x_1,\ldots,x_N$.

\begin{remark}[Practical Considerations]
    In practice, we can fix just one random coloring $\Delta$ of the dictionary rather than generate an independent random coloring $\Delta_j$ for every next token. Independent colorings help us establish theoretical guarantees, but they could be too conservative in practice.
   
\end{remark}

\subsection{Validity and Undetectability}

Now let's check that the watermarking scheme we described is valid (i.e. $q$ is actually a probability distribution) and undetectable. 

\begin{proposition}[Validity]
    The watermarked distribution $q$ defined in \eqref{eq:vector_q_standard} is indeed a probability distribution over $[d]$.
\end{proposition}

\begin{proof} 
By design, every token $i\in I$ is matched to some $j\in I^c$ such that $\varepsilon(i)=\varepsilon(j)$. Since $\Delta(i)$ and $\Delta(j)$ have opposite signs, we have
$\Delta(i)\varepsilon(i) = -\Delta(j)\varepsilon(j)=0$. Thus, when we sum all these terms, they cancel:  
\begin{equation*}
\sum_{i=1}^d\varepsilon(i)\Delta(i)= \sum_{i\in I}\Delta(i)\varepsilon(i) +\sum_{j\in I^c}\Delta(j)\varepsilon(j) =0,
\end{equation*}
and we get
\begin{equation*}
\sum_{i=1}^d q(i) = \sum_{i=1}^dp(i) + r\sum_{i=1}^d\varepsilon(i)\Delta(i) = 1.
\end{equation*}
So, to check that $q$ defines a probability distribution, it remains to show that it has nonnegative entries. But this follows from our construction: 
\begin{equation*}
q(i)\ge p(i)-\varepsilon(i) = p(i)-\min \left( w(i'),\overline{w}(i') \right) \ge p(i)-p(i)=0. \qedhere
\end{equation*}
\end{proof}

\begin{proposition}[Undetectability]
\label{prop:undetectability_closed}
    The watermarked distribution $q$ defined in \eqref{eq:vector_q_standard} is indeed a probability distribution over $[d]$, and the watermarking satisfies the undetectability requirement.
\end{proposition}

\begin{proof}
For every $i\in [d]$, the probability that $x=i$ under $q$ is
\begin{equation*}
\begin{split}
&\mathbb{P}_q\{x=i\} = (p(i)+r\varepsilon(i))\Prob{\Delta(i)=1} + (p(i)-r\varepsilon(i))\Prob{\Delta(i)=-1} \\
&= (p(i)+r\varepsilon(i))\frac{1}{2} + (p(i)-r\varepsilon(i))\frac{1}{2} = p(i). 
\end{split} 
\end{equation*}
Because at each step $j=1,\ldots,N$ we sample an independent copy of $\Delta$, the distribution of the text $x_1,\ldots,x_N$ sampled from the watermarked distributions $q$ and the unwatermarked distributions $p$ remains the same.
\end{proof}

\subsection{Soundness and Completeness
}
We now describe how our detection algorithm works. The idea is that whenever $r\Delta(x)=1$ is positive,  the chance of the token $x$ to appear in the text is increased, and when $r\Delta(x)=1$ the chance is decreased. 

Thus, the watermarked text should contain more tokens $x$ for which $r\Delta(x)=1$ than those for which $r\Delta(x)=-1$. This is testable: just check if the empirical mean of the random variables $r_j\Delta_j(x_j)$ significantly exceeds zero.

On the other hand, an non-watermarked text is independent of the key, so it should contain roughly the same number of tokens $x$ for which $r\Delta(x)=1$ as those for which $r\Delta(x)=-1$. Thus, a non-watermarked text has a negligible chance to pass the test above.

Let's make this all precise.

\begin{algorithm}
\caption{Watermark Detection}
\label{algo:detection_standard}
\begin{algorithmic}
\Require The text $x_1,\ldots,x_{N}$. The secret key: $r_1\Delta_1,\ldots,r_N\Delta_N$. Tolerance  $\delta$.
\Ensure True or False.
\smallskip

\noindent Compute $Z = \frac{1}{N}\sum_{j=1}^N r_j\Delta_j\big(x_j\big)$. 

\smallskip

\noindent If $Z\ge \sqrt{2\log(1/\delta)/N}$ return True, else return False.

\end{algorithmic}
\end{algorithm}

\begin{proposition}[Soundness]
\label{prop:soundness_closed}
The watermarking scheme described in Section \ref{s: watermarking scheme closed} is sound probability at least $1-\delta$. 
\end{proposition}

\begin{proof}
    Suppose we are under the null hypothesis $H_0$ that the text $x_1,\ldots,x_N$ is independent of the secret key $r_1\Delta_1,\ldots,r_N\Delta_N$. Then the test $Z$ defined in Algorithm \ref{algo:detection_standard} is the normalized sum of independent bounded random variables. By Hoeffding's inequality, for every $t>0$ we have
    \begin{equation*}
        \Prob{Z\ge t}\le e^{-N t^2/2},
    \end{equation*}
    which is at most $\delta$ for $t=\sqrt{2\log(1/\delta)/N}$.
\end{proof}

\begin{theorem}[Completeness]
\label{thm:completness_closed}
    Let $\delta \in (0,1)$ and $N\ge 2\log(1/\delta)$. Let $p^*_j$ denote the probability of the most likely token in the dictionary at step $j=1,\ldots, N$. If, for some  $\gamma >0$,
    \begin{equation} \label{eq:main_standard}
        \frac{1}{20N}\sum_{i=1}^N(1-p^*_j)\ge \sqrt{\frac{2\log(1/\delta)}{N}} + \gamma,
    \end{equation}
    then the watermarking scheme described in Section \ref{s: watermarking scheme closed} is complete with probability at least $1-e^{-N\gamma^2/2}$.
\end{theorem}

\begin{remark}[Minimal entropy allows watermarking]
\label{remark:minimal_entropy}
    The assumption \eqref{eq:main_standard} on the distributions of the tokens is nearly necessary. If $p_j^{\ast}$ gets too close to $1$, it means that the token is almost deterministic, so watermarking it is impossible. The assumption \eqref{eq:main_standard} ensures that at least some fraction of the text is non-deterministic, making watermarking possible in principle. 
\end{remark}

Before we proceed to the proof of Theorem \ref{thm:completness_closed}, we require one preliminary lemma. 
\begin{lemma}[Maxima of subsets]
\label{lemma:nonincreasing_rearrangments}
    Consider any numbers $p_1\ge \cdots \ge p_d\ge 0$. Let $I$ be a random subset of $[d]$ with cardinality $d/2$, uniformly sampled from all such subsets. Denote by $w_i$ the $i$-th largest element of $\{p_i\}_{i\in I}$, and by $\overline{w}_i$ the $i$-th largest element of $\{p_i\}_{i\in I^c}$. Then
    \begin{equation*}
        \frac{1}{20} \sum_{i=2}^d p_i
        \le \E \sum_{i=1}^{d/2} \min(w_i, \overline{w}_i)  
        \le \sum_{i=2}^d p_i.
    \end{equation*}
\end{lemma}

We prove this lemma in Appendix \ref{s: proof maxima}.

\begin{proof}[Proof of Theorem \ref{thm:completness_closed}]
Assume that the text is watermarked. Thus, at each  step, given the previous tokens, the random token $x$ is picked from the distribution $q$ defined in  \eqref{eq:vector_q_standard}. Then
$$
\E r\Delta(x)
= \E \sum_{i=1}^d r\Delta(i) \one_{\{x=i\}}
= \E \sum_{i=1}^d r\Delta(i)\E \Big[ \one_{\{x=i\}} \vert \Delta \Big].
$$
By \eqref{eq:vector_q_standard}, we have 
$$
\E \Big[ \one_{\{x=i\}} \vert \Delta \Big]
= \Prob{x=i \given \Delta}
= p(i) + r\e(i) \Delta(i).
$$
So, using that $r\Delta(i)$ has Rademacher distribution, we get 
$$
\E r\Delta(x)
= \sum_{i=1}^d p(i) \underbrace{\E r\Delta(i)}_{=0} + \E \sum_{i=1}^d \e(i) \underbrace{r^2\Delta(i)^2}_{=1}
= \E \sum_{i=1}^d \e(i).
$$
Now, using the definition of $\e(i)$ in \eqref{eq: epsilons} and applying Lemma \ref{lemma:nonincreasing_rearrangments}, we get 
$$
\E r\Delta(x)
= \E \sum_{i=1}^{d/2} \min \big( w(i), \overline{w}(i) \big) 
\ge \frac{1}{20}\big(1-p^*\big),
$$
where $p^* = \max(p(1), \ldots, p(d))$. repeating this argument at each step $j=1,\ldots,N$ conditionally on the previous steps, we get 
\begin{equation}
\label{eq:conditional_expectation_token_closed}
\E[r_j\Delta_j(x_j)|x_1,\ldots,x_{j-1}]\ge \frac{1}{20}\big(1-p^{\ast}_j\big).
\end{equation}
Next, set $Z_{0}\coloneqq0$ and notice that
\begin{equation*}
Z_{j}\coloneqq Z_{j-1} + r_j\Delta_j(x_j)-\E[r_j\Delta_j(x_j)|x_1,\ldots,x_{j-1}],
\end{equation*}
is a martingale satisfying that $|Z_{k}-Z_{k-1}|\le 1$ almost surely. Thus, by combining the Azuma-Hoeffding's inequality and \eqref{eq:conditional_expectation_token_closed} we obtain
\begin{equation*}
    \Prob*{Z\ge \frac{1}{20N}\sum_{i=1}^N(1-p^*_j)-\gamma}\ge 1-e^{-N\gamma^2/2}.
\end{equation*}
Now it remains to use  assumption \eqref{eq:main_standard}.
\end{proof}

\section{Open Setting}
\label{sec:open}
\subsection{The watermarking scheme}
\label{sec: watermarking scheme open}
Unfortunately, in the open source setting, there is no guarantee that an adversary cannot learn our watermarking scheme. So we propose a completely different approach in this setting. 

We randomly perturb the logits in \eqref{eq:softmax_rule}, similarly to \cite{christ2024provably}. However, the perturbation vector in \cite{christ2024provably} was sampled from a multivariate Gaussian distribution (the same at each step), while we propose to draw perturbation vectors from a random Gaussian mixture, which changes independently at each step. 

Let's describe the construction of the perturbation vector $\Delta \in \mathbb{R}^d$ in detail. First, pick a $k$-sparse subset $S\subset [d]$ uniformly at random, and compute the unit norm vector 
$$
\mu = k^{-1/2}\mathds{1}_S
$$
supported on $S$. Let $r$ be a Rademacher random variable. For a fixed $\varepsilon>0$, let $G\sim N(0,\varepsilon^2 I)$ and define the perturbation vector
\begin{equation}
\label{eq:mixture_Gaussian}
\Delta\coloneqq
\begin{cases}
    G + \varepsilon \mu, \quad \text{if $r=1$}\\
    G -\varepsilon \mu, \quad \text{if $r=-1$.}
\end{cases}
\end{equation}
The role of the tuning parameter  $\varepsilon$ is to control the tradeoff between detectability of the watermark and the quality of the text. The {\em secret key} is the collection of vectors $G_1,\ldots,G_{N}$ used to generate the i.i.d. copies $\Delta_1,\ldots, \Delta_N$ of $\Delta$ at each step according to \eqref{eq:mixture_Gaussian}.

At each new step, the LLM computes the logits $\ell(1),\ldots,\ell(d)$, and our watermarking algorithm samples the new token according to the watermarked softmax rule given by
\begin{equation}
\label{eq:vector_q_open}
    q(i)\coloneqq \frac{e^{\ell(i)+\Delta(i)}}{e^{\ell(1)+\Delta(1)}+\cdots+e^{\ell(d)+\Delta(d)}},
\end{equation}
which is just a perturbed version of \eqref{eq:softmax_rule}. 
We allow the adversary to have knowledge of the distribution $q$.
\subsection{Unremovability}
Intuitively, in order to remove the watermarking scheme, the adversary needs to guess the values of the perturbation vectors $\Delta_1,\ldots,\Delta_N$. This requires the adversary to learn the mean $\mu$ accurately.

But the mixture distribution is chosen exactly for the purpose of making the adversary's task computationally hard. Indeed, estimating the mean $\mu$ based on the samples $\Delta_1,\ldots,\Delta_N$ is a well-known computationally hard problem in robust statistics, called {\em sparse mean estimation under Huber's contamination noise}.

The choice of the mixture distribution allow us to exploit the following hypothesis testing version of the sparse mean estimation problem.
\medskip

\textbf{Sparse Mean Hypothesis Testing:}
\begin{equation}
\label{eq:sparse_mean_hypotest}
\begin{aligned}
\text{\tiny$\bullet$}\quad H_0 &: \Delta_1,\ldots,\Delta_N \sim N(0,\varepsilon^2 I) \\
\text{\tiny$\bullet$}\quad H_a &: \Delta_1,\ldots,\Delta_N \sim \text{mixture in } \eqref{eq:mixture_Gaussian}
\end{aligned}
\end{equation}

Clearly, if it is impossible to distinguish $H_0$ and $H_a$ in the sparse mean hypothesis testing, then it is not possible to estimate accurately the mean $\mu$ (or $-\mu$) based on $\Delta_1,\ldots,\Delta_N$. 

Perhaps surprisingly, Brennan and Bresler \cite{brennan2020reducibility} showed that under a well-known conjecture in theoretical computer science, {\em the $k$-BPC conjecture}, solving the sparse mean estimation problem can be hard:

\begin{theorem}[Brennan and Bresler \cite{brennan2020reducibility}]
\label{thm:BB_impossibilit}
In the sparsity regime $k\ll \sqrt{d}$, no polynomial-time algorithm on $N$ can solve \eqref{eq:sparse_mean_hypotest} with less than $N =\Tilde{\Omega}(k^2)$ samples, assuming the $k$-BPC conjecture is true.  On the other hand, there exists a computationally infeasible algorithm that solves \eqref{eq:sparse_mean_hypotest} with $k=\Theta(k\log d)$ samples.
\end{theorem}

The gap between $k$ and $k^2$ (hiding log factors) is an example of statistical-to-computational gaps, a topic extensively studied in theoretical computer science literature (see \cite{kunisky2019notes} and the references therein).

Impossibility results in machine learning, statistics, and computer science, which result in statistical-to-computational gaps, are usually interpreted as ``negative'' statements. Here, our perspective is different: we leverage a statistical-to-computational gap to our advantage -- to safeguard our watermarking from adversarial attacks.

We now prove the main fact about unremovability:
\begin{proposition}[Unremovability]
\label{prop:unremovability}
Assume that $ d^{\delta} \ll N^{1+\delta} \ll d$ for some fixed $\delta>0$. Let $k = N^{(1+\delta)/2}$. Then, at each step, any Gaussian distribution the adversary can learn in polynomial time has TV distance at least $0.4$ from the distribution of the watermarked logits $\ell(i) + \Delta(i)$ in \eqref{eq:vector_q_open}.
\end{proposition}
\begin{proof}
The best the adversary can hope for is to predict $\Delta_{adv} \sim N(\mu_{adv},\varepsilon^2I)$, for some $\mu_{adv}$ satisfying $\|\mu-\mu_{adv}\|_2\ge \varepsilon/2$. (Indeed, if $\|\mu_{adv}-\mu\|_2<\varepsilon/2$, then the adversary would solve the sparse mean estimation hypothesis testing which is not possible in polynomial time thanks to Theorem \ref{thm:BB_impossibilit}. Indeed, to learn the mean $\mu$ with higher accuracy one would need more than $\Tilde{\Theta}(k^2)\gg N$ samples.)

Notice that the total variation distance between the $\Delta_{adv}$ and $\Delta$ is at least $\Phi(-1/4)>0.4$, where $\Phi(\cdot)$ is the cumulative density function of a standard multivariate Gaussian. 
\end{proof}

The reason why we need to assume the regime $d^{\delta}\ll N^{1+\delta}\ll d$ is to fulfill the hypothesis of the Theorem \ref{thm:BB_impossibilit} due to the log factors in Theorem \ref{thm:BB_impossibilit}.

Also note that the choice of total variation distance is not essential, we just require a distance between Gaussian distributions that is bounded away from zero if their means are $1/2$-far apart in Euclidean distance. For example, similar guarantees holds for KL-divergence or $2$-Wassertein distance.

\subsection{Soudness and Completness}
We describe our watermarking detection algorithm. The idea is similar to the one we used in the closed setting: if the text is independent from the watermark, then $G_i(x_i)$ are distributed as standard Gaussians and therefore the empirical mean concentrates around $0$. On the other hand, if the text is generated by the watermarking scheme then we have a bias towards the positive entries of $G_1,\ldots,G_N$ and consequently the empirical mean should deviate from zero.

We now describe the detection algorithm for the open source setting:
\begin{algorithm}
\caption{Watermark Detection for the Open Source Setting}
\label{algo:detection_open}
\begin{algorithmic}
\Require The text $x_1,\ldots,x_{N}$. The key: $G_1,\ldots,G_N$. Tolerance $\delta$.
\Ensure True or False.

\smallskip

\noindent Compute $Z = \frac{1}{N}\sum_{j=1}^N  G_j(x_j)$.

\smallskip

\noindent If $Z \ge \varepsilon\sqrt{2\log(1/\delta)/N}$ return True, else return False.
\end{algorithmic}
\end{algorithm}

\medskip

\begin{proposition}[Soundness]
\label{prop:soundness_open}
The watermarking scheme described in Section \ref{sec: watermarking scheme open} is sound with probability at least $1-\delta$. 
\end{proposition}
\begin{proof}
Notice that under the null hypothesis, $x_1,\ldots,x_N$ are independent from $G_1,\ldots,G_N$, therefore the test $Z$ defined in Algorithm \ref{algo:detection_open} is the distributed as $Z\sim N(0,\varepsilon^2/N)$. By the standard estimate for the Gaussian tail,
\begin{equation*}
    \Prob{Z \ge \varepsilon t}\le e^{-t^2N/2},
\end{equation*}
which is at most $\delta$ for $t=\sqrt{2\log(1/\delta)/N}$.
\end{proof}

\noindent Let $c_0$ be the absolute constant in \cite[Proposition 2.7.1 from (b) to (e)]{vershynin2018high}. Our main result for the open source setting is the following theorem:
\begin{theorem}[Completeness]
\label{thm:completeness_open}
Let $\delta \in (0,1)$, $N\ge 2\log(1/\delta)$. Let $c_0$ as above and let $p^*_j$ denote the probability of the most likely token in the dictionary at step $j=1,\ldots, N$. If, for some $\gamma>0$ and $\varepsilon\le 1/2$,
\begin{equation}
\label{eq:main_open}
        \frac{\varepsilon^2}{1200}\frac{1}{N}\sum_{j=1}^{N}(1-p_j^{\ast}) \ge \varepsilon \sqrt{\frac{2\log(1/\delta)}{N}}+\gamma.
\end{equation}
Then, with probability at least $1- e^{-\gamma^2N/672c_0^2\varepsilon^2}$, the watermark scheme described in Section \ref{sec: watermarking scheme open} is complete.
\end{theorem}
We remark that the assumption \eqref{eq:main_open} is analogous to assumption \eqref{eq:main_standard} and it is somewhat necessary as explained in the Remark \ref{remark:minimal_entropy}.

We made effort to make all constants explicit, but we opt for a more simplified analysis rather than optimizing the value of the constants. Finally, the assumption on $\varepsilon\le 1/2$ is for technical convenient and could be replace by any absolute constant if necessary.

Before we proceed to the prove, we require some preliminary results. Recall the softmax functions \eqref{eq:softmax_rule} and \eqref{eq:vector_q_open}.

\begin{proposition}
\label{prop:fundametal_bias}
Let $G=(g_1,\ldots,g_d) \sim N(0,\varepsilon^2I)$, for some $\varepsilon \le 1/2$ and $x \in [d]$ be a token sampled according to the watermarked softmax rule \eqref{eq:vector_q_open}. Set $p^{\ast}$ to be the non-increasing rearrangement of the unwatermarked probability distribution $p$ \eqref{eq:softmax_rule} of the token $x$. Then
\begin{equation*}
    \E G(x) \ge  \frac{\varepsilon^2}{1200}(1-p^{\ast}(1)).
\end{equation*}
\end{proposition}

\begin{lemma}
\label{psi1_estimate}
Let $G=(g_1,\ldots,g_d)\sim N(0,\varepsilon^2I)$, for some $\varepsilon\le1/2$ and $x \in [d]$ be a token sampled according to the watermarked softmax rule \eqref{eq:vector_q_open}. Then the random variable $G(x)$ is sub-exponential with $\|G(x)\|_{\psi_1}\le 2.8\sqrt{10}\varepsilon$ and  $\|G(x)-\E G(x)\|_{\psi_1}\le 8.4\sqrt{10}\varepsilon$.
\end{lemma}
We leave the proof of Proposition \ref{prop:fundametal_bias} to Appendix \ref{sec:appendix_prop_fundamental_bias} and the proof of Lemma \ref{psi1_estimate} to Appendix \ref{sec:proof_psi1_estimate}. 

\begin{proof}[Proof of Theorem \ref{thm:completeness_open}]
Assume that the text is watermarked. Set $Z_{0}\coloneqq 0$ and notice that
\begin{equation*}
Z_{j}\coloneqq Z_{j-1} +  G_k(x_j) - \E \big[ G_j(x_j)|x_1,\ldots,x_{j-1}\big],
\end{equation*}
is a martingale. By Lemma \ref{psi1_estimate}, the increments $$Y_j \coloneqq [Z_{j}-Z_{j-1}]|x_{1},\ldots,x_{j-1},$$ are sub-exponential. In addition to this,  \cite[Proposition 2.7.1]{vershynin2018high} implies that there exists a constant $c_0>0$ for which all the sub-exponential random variables $Y_j$ satisfy
\begin{equation*}
\E e^{\lambda |Y_j|}\le e^{\lambda^2c_0^2\|Y_j\|_{\psi_1}^2} \quad \text{for all} \quad |\lambda|\le \frac{1}{c_0 \|Y_j\|_{\psi_1}}.
\end{equation*}
It follows from the sub-exponential version of the Azuma-Hoeffding's inequality \cite[Theorem 2.3]{wainwright2019high} and Proposition \ref{prop:fundametal_bias} that 
\begin{equation*}
\begin{split}
&\mathbb{P}\bigg\{Z\ge \frac{\varepsilon^2}{1200}\frac{1}{N}\sum_{j=1}^N(1-p_j^{\ast})-t\bigg\}\ge 1-e^{-t^2N/2c_0^284\varepsilon^2}.
\end{split}
\end{equation*}
The result follows from $t=\gamma/2$ combined with the assumption \eqref{eq:main_open}.

\end{proof}

\appendix 

\section{Proof of Lemma \ref{lemma:nonincreasing_rearrangments}} \label{s: proof maxima}
We start with the proof of the lower bound. By the Hoeffding-Chernoff inequality for the hyper-geometric distribution (see \cite[Section 27.5]{frieze2015introduction}), for each $i=2,\ldots,d/10$, we have:
\begin{equation}
\label{eq:balanced_condition}
\big|I\cap [10i]\big|\ge i \quad \big|I^c\cap [10i]\big|\ge i,
\end{equation}
with probability at least $1-2e^{-8i/5}$. Notice that for any $i$ for which the condition \eqref{eq:balanced_condition} holds, the monotonicity of $p$ implies that at least $i$ elements of $\{p(i)\}_{i\in I}$ are greater or equal to $p_{10i}$. By definition of $w$, this implies that $w_i\ge p_{10i}$. Similarly, we have $\overline{w}_i\ge p_{10i}$ which implies $\min\{w_i,\overline{w}_i\}\ge p_{10i}$. 

Since for each $i=2,\ldots,d/10$, the event \eqref{eq:balanced_condition} holds with probability at least $1-2e^{-8i/5}\ge 1/2$, we have 
\begin{equation*}
    \E \min\{w_i,\overline{w}_i\}\ge \frac{p_{10i}}{2}.
\end{equation*}
Additionally, with probability $1/2$ we have $1\in I$ and $2\in I^c$ (vice-versa), so
\begin{equation*}
    \E\min\big\{w_1,\overline{w}_1\big\}\ge \frac{p_2}{2}.
\end{equation*}
We conclude that 
\begin{equation}
\label{eq:sparse_sum}
\E\sum_{i=1}^{d/2}\min\{w(i),\overline{w}(i)\} \ge \frac{1}{2}\big(p_2+p_{10}+p_{20}+\cdots + p_d\big)
\end{equation}
To complete the full sum (from $i=2$ to $d$), notice that
\begin{equation}
\label{eq:filling_sum}
8p_2\ge p_3+\cdots+p_9 \quad \text{and} \quad 10p_{10i}\ge p_{10i}+\cdots+p_{10i+9}.
\end{equation}
Combining \eqref{eq:sparse_sum} and \eqref{eq:filling_sum}, we obtain
\begin{equation*}
\E\sum_{i=1}^{d/2}\min\{w_i,\overline{w}_i\} \ge \frac{1}{20}\sum_{i=2}^d p_i = \frac{1}{20}(1-p_1).
\end{equation*}
The proof of the upper bound is shorter. In fact, almost surely we have
\begin{equation*}
\begin{split}
&\sum_{i=1}^d\min\big\{w_i,\overline{w}_i\big\} \le \min\bigg\{\sum_{i=1}^{d/2}w_i,\sum_{i=1}^{d/2}\overline{w}_i\bigg\}\\
&\le \min\bigg\{\sum_{i\in I}p_i,\sum_{i\in I^c}p_i\bigg\} \le \sum_{i=2}^dp_i.
\end{split}
\end{equation*}
where the last inequality holds because either $I$ or $I^c$ does not contain the element $1$.
The proof is complete.

\section{Proof of Proposition \ref{prop:fundametal_bias}}
\label{sec:appendix_prop_fundamental_bias}
To start, we focus on the term 
\begin{equation*}
\E \bigg[g_k \frac{e^{\Delta(k)+\ell(k)}}{e^{\ell(1)+\Delta(1)}+\cdots + e^{\ell(d)+\Delta(d)}}\bigg] \ge e^{-\varepsilon/\sqrt{k}} \E \bigg[g_k \frac{e^{g_k}}{e^{\ell(1)+g_1-\ell(k)}+\cdots + e^{\ell(d)+g_d-\ell(k)}}\bigg].
\end{equation*}
Clearly, $e^{-\varepsilon/\sqrt{k}}\ge e^{-1}$ as $k\ge 1$. Next, denote by $\E_j$  the expectation with respect to the randomness of $g_j$ only. By the iterated law of expectation, we compute the expectation term in the right-hand side by
\begin{equation*}
 \E_{1,\ldots,k-1,k+1,\ldots,d} \bigg[\E_k \bigg[g_k\frac{e^{g_k}}{e^{\ell(1)+g_1-\ell(k)}+\ldots + e^{\ell(d)+g_d-\ell(k)}}\bigg]\bigg],
\end{equation*}
and by independence we may treat $\sum_{i\neq k}e^{\ell(i)-\ell(k)+g_i}$ as a constant for the inner expectation. Thus, let us define $\alpha_k \coloneqq \sum_{i\neq k}e^{\ell(i)-\ell(k)+g_i}$ and write
\begin{equation*}
\begin{split}
&\E_k \bigg[g_k \frac{e^{g_k}}{\alpha_k + e^{g_k}}\bigg] = \frac{1}{2}\E_{k}\bigg[|g_k| \frac{e^{|g_k|}}{\alpha_k + e^{|g_k|}} - |g_k| \frac{e^{-|g_k|}}{\alpha_k + e^{-|g_k|}}\bigg]\\
&= \frac{\alpha_k}{2}\E_k \bigg[|g_k| \bigg(\frac{e^{|g_k|}-e^{-|g_k|}}{\alpha_k^2+\alpha_k(e^{|g_k|}+e^{-|g_k|})+1}\bigg)\bigg].
\end{split}
\end{equation*}
Next, notice that the function $f: [0,\infty)\rightarrow \mathbb{R}$
\begin{equation*}
    f(x) \coloneqq \frac{e^{x}-e^{-x}}{\alpha_k^2+\alpha_k(e^{x}+e^{-x})+1},
\end{equation*}
is increasing and non-negative. Invoking the FKG inequality, we have that
\begin{equation*}
\begin{split}
&\E_k \bigg[g_k \frac{e^{g_k}}{\alpha_k + e^{g_k}}\bigg] \ge \frac{\varepsilon\alpha_k}{\sqrt{2\pi}} \E_k\bigg[\frac{e^{|g_k|}-e^{-|g_k|}}{\alpha_k^2+\alpha_k(e^{|g_k|}+e^{-|g_k|})+1}\bigg] \\
&\ge \frac{\varepsilon\alpha_k}{\sqrt{2\pi}} \left(\frac{e^{\varepsilon c}-e^{-\varepsilon c}}{\alpha_k^2+\alpha_k(e^{\varepsilon c}+e^{-\varepsilon c})+1}\right)\mathbb{P}\{|g_k|\ge \varepsilon c\}.
\end{split}
\end{equation*}
By the standard tail estimate for Gaussians, for the choice of $c=1/2$, it follows that $\Prob{|g_k|\ge \varepsilon/2}\ge 1/2$. Recall that $\varepsilon\le 1/2$ which implies that $e^{\varepsilon/2}+e^{-\varepsilon/2}\le 2.1$, thus 
\begin{equation*}
   \alpha_k^2+\alpha_k(e^{\varepsilon/2}+e^{-\varepsilon/2})+1 \le \frac{2.06}{2}\left(\alpha_k^2 + 2\alpha_k+1\right)\le 1.05\left(\alpha_k^2 + 2\alpha_k+1\right),
\end{equation*}
and 
\begin{equation*}
\begin{split}
&\E_k g_k \bigg[\frac{e^{g_k}}{\alpha_k + e^{g_k}}\bigg] \ge \frac{\varepsilon\alpha_k}{\sqrt{8\pi}} \left(\frac{e^{\varepsilon/2}-e^{-\varepsilon/2}}{\alpha_k^2+\alpha_k(e^{\varepsilon/2}+e^{-\varepsilon/2})+1}\right)\\
&\ge \varepsilon\frac{e^{\varepsilon/2}-e^{-\varepsilon/2}}{1.05\sqrt{8\pi}} \left(\frac{\alpha_k}{\alpha_k^2+2\alpha_k+1}\right)\\
&\ge \frac{\varepsilon^2}{1.05\sqrt{8\pi}} \left(\frac{\alpha_k}{\alpha_k^2+2\alpha_k+1}\right) = \frac{\varepsilon^2}{1.05\sqrt{8\pi}} \frac{\alpha_k}{(\alpha_k+1)^2}.
\end{split}
\end{equation*}
It remains to estimate (from below)
\begin{equation}
\label{eq: second_main_term}
   \frac{\varepsilon^2}{1.05\sqrt{8\pi}}\sum_{k=1}^d \E_{1,\ldots,k-1,k+1,\ldots,d} \bigg[\frac{\alpha_k}{(\alpha_k+1)^2}\bigg].
\end{equation}
Or equivalently (up to the multiplicative constant in front),
\begin{equation*}
\begin{split}
\sum_{k=1}^d \E_{1,\ldots,k-1,k+1,\ldots,d} \bigg[\frac{1}{(\alpha_k+1)} - \frac{1}{(\alpha_k+1)^2}\bigg].
\end{split}
\end{equation*}
Recalling that $\alpha_k = \sum_{i\neq k}e^{\ell(i)-\ell(k)+g_i}$, suppose that there is a non-empty event $\mathcal{E}$ for which both conditions below hold simultaneously
\begin{equation}
\label{eq:event_E}
    \sum_{i\neq k}e^{\ell(i)-\ell(k)+g_i} \le 4.55 \sum_{i\neq k}e^{\ell(i)-\ell(k)} \quad \text{and} \quad \sum_{i\neq k}e^{\ell(i)-\ell(k)+g_i}\ge \frac{1}{4}\sum_{i\neq k}e^{\ell(i)-\ell(k)}.
\end{equation}
By the second estimate in \eqref{eq:event_E}, we have that
\begin{equation*}
    \frac{1}{\alpha_k}=\frac{e^{\ell(k)}}{\sum_{i\neq k}e^{\ell(i)+g_i}} \le 4\frac{e^{\ell(k)}}{\sum_{i\neq k}e^{\ell(i)}}  = \frac{4p(k)}{1-p(k)},
\end{equation*}
and then 
\begin{equation*}
    \frac{1}{(1+\alpha_k)^2}\le \frac{1}{(1+\alpha_k)(1+(1-p(k))/4p(k))}.
\end{equation*}
Consequently, 
\begin{equation*}
    \frac{1}{1+\alpha_k}-  \frac{1}{(1+\alpha_k)^2} \ge \bigg(\frac{1}{1+\alpha_k}\bigg)\frac{1-p(k)}{4p(k)+1-p(k)}\ge \bigg(\frac{1}{1+\alpha_k}\bigg)\frac{1-p^{\ast}(1)}{4}.
\end{equation*}
Finally, notice that the first estimate in \eqref{eq:event_E} implies that
\begin{equation*}
\begin{split}
& \sum_{k=1}^d \frac{1}{1+\alpha_k} = \sum_{k=1}^d \frac{e^{\ell(k)}}{e^{\ell(k)}+4.55\sum_{i\neq k}e^{\ell(i)}} \\
&\ge \frac{1}{4.55}\sum_{k=1}^d \frac{e^{\ell(k)}}{e^{\ell(k)}+\sum_{i\neq k}e^{\ell(i)}} = \frac{1}{4.55} \ge 0.2.
\end{split}
\end{equation*}
Since \eqref{eq: second_main_term} is non-negative, we have that
\begin{equation*}
\sum_{k=1}^d \E_{1,\ldots,k-1,k+1,\ldots,d} \bigg[\frac{\alpha_k}{(\alpha_k+1)^2}\bigg] \ge 0.05 (1-p^{\ast}(1))\mathbb{P}\{\mathcal{E}\}
\end{equation*}
All it remains is to prove that $\mathbb{P}\{\mathcal{E}\}$ is bounded away from zero. To this end, notice that by Markov's inequality
\begin{equation*}
   \Prob{\alpha_k \le 4\E\alpha_k} \ge \frac{3}{4}.
\end{equation*}
Next, notice that for every $i\neq k$,
\begin{equation*}
  \frac{e^{\ell(i)}}{\sum_{i\neq k}e^{\ell(i)}}e^{g_i}\ge \frac{e^{\ell(i)}}{\sum_{i\neq k}e^{\ell(i)}} \mathds{1}_{g_i\ge 0}\eqqcolon a_i \mathds{1}_{g_i\ge 0}.
\end{equation*}
The random variable $Y\coloneqq\sum_{i\neq k}a_i\mathds{1}_{g_i\ge 0}$ is the sum of independent non-negative random variables $Y_i$, where $Y_i\in [0,a_i]$. Notice that $\sum_{i\neq k}a_i=1$, thus we split into two cases. If there exists an $a_i\ge 1/4$ then with probability $1/2$, $Y\ge 1/4$. On the other hand, if $a_i\le 1/4$ for all $i\neq k$ then by Hoeffding's inequality
\begin{equation*}
    \Prob*{Y\ge \frac{1}{2} - t }\ge 1- e^{-32t^2}.
\end{equation*}
Setting $t=1/4$, we obtain that
\begin{equation*}
    \Prob*{Y\ge \frac{1}{4}}\ge 1-e^{-2}.
\end{equation*}
We conclude that $\Probevent{\mathcal{E}} \ge \min\{1/2-e^{-2},3/4-1/2\} = 1/4$, which finishes the proof.

\section{Proof of Lemma \ref{psi1_estimate}}
\label{sec:proof_psi1_estimate}
Since $\|\cdot\|_{\psi_1}$ is a norm, we have that 
\begin{equation*}
    \|G(x) - \E G(x)\|_{\psi_1} \le \|G(x)\|_{\psi_1} + \|\E G(x)\|_{\psi_1} =\|G(x)\|_{\psi_1} + |\E G(x)| \le 3\|G(x)\|_{\psi_1},
\end{equation*}
where the last step follows from 
\begin{equation*}
\left|\frac{1}{\|G(x)\|_{\psi_1}}\E G(x)\right| \le \left|1+\frac{1}{\|G(x)\|_{\psi_1}}\E G(x)\right|\le \E e^{|G(x)|/\|G(x)\|_{\psi_1}} \le 2.
\end{equation*}
The proof boils down to showing that for some well-chosen $\tau$ (small as possible)
\begin{equation*}
    \E e^{|G(x)|/\tau} \le e^{\varepsilon/\sqrt{k}} \sum_{k=1}^d \E\bigg[ e^{g_k/\tau}\frac{e^{g_k}}{e^{\ell(1)-\ell(k)+g_1}+\ldots+e^{\ell(d)-\ell(k)+g_d}}\bigg]  \le 2,
\end{equation*}
which implies that $\|G(x)\|_{\psi_1}\le \tau$. We proceed similarly as in Proposition \ref{prop:fundametal_bias}. Let $\alpha_k \coloneqq \sum_{i\neq k}e^{\ell(i)-\ell(k)+g_i}$ and notice that by Cauchy-Schwarz inequality,
\begin{equation*}
\begin{split}
&\sum_{k=1}^d \E_k \bigg[e^{g_k/\tau}\frac{e^{g_k}}{e^{\ell(1)-\ell(k)+g_1}+\ldots+e^{\ell(d)-\ell(k)+g_d}}\bigg]\\
&\le \sum_{i=1}^d (\E_k e^{2g_k/\tau})^{1/2} \left(\E_k\bigg[\frac{e^{2g_k}}{(\alpha_k+e^{g_k})^2}\bigg]\right)^{1/2}.
\end{split}
\end{equation*}
We claim that 
\begin{equation*}
 \left(\E_k\bigg[\frac{e^{g_k}}{(\alpha_k+e^{g_k})^2}\bigg]\right)^{1/2} \le 2.2 \E_k\bigg[\frac{e^{\ell(k)+g_k}}{e^{\ell(1)+g_1}+\ldots+e^{\ell(d)+g_d}}\bigg].
\end{equation*}
If the claim is true then by the law of iterated expectation
\begin{equation*}
\E e^{|G(x)|/\tau} \le 2.2e(\E^{2g/\tau})^{1/2}\E\bigg[\sum_{k=1}^d \frac{e^{\ell(k)+g_k}}{e^{\ell(1)+g_1}+\ldots+e^{\ell(d)+g_d}}\bigg] \le 6e^{\varepsilon^2/\tau^2}.
\end{equation*}
Choosing $\tau \ge \sqrt{10}\varepsilon$, we reach the estimate $\E e^{|G(x)|/\tau} \le 6.7$. We would like to replace the constant $6.7$ by $2$ in the right-hand side. To this end, notice that for the constant $a=\log 6.7/\log 2>1$, the function $f(x)=x^{1/a}$ is concave and then Jensen inequality implies that
\begin{equation*}
\E e^{|G(x)|/a\tau} = \E \big[(e^{|G(x)|/\tau})^{1/a}\big] \le \left(\E e^{|G(x)|/\tau}\right)^{1/a}\le (6.7)^{1/a} =2.
\end{equation*}
Thus setting $\tau = 2.8\sqrt{10}\varepsilon$ concludes the proof.
We now proceed to prove the claim. At one hand,
\begin{equation}
\label{eq:without_square}
\begin{split}
&\E\bigg[\frac{e^{g_k}}{\alpha_k+e^{g_k}}\bigg] = \frac{1}{2}\E \bigg[\frac{e^{-|g_k|}}{\alpha_k+e^{-|g_k|}}\bigg] + \frac{1}{2}\bigg[\frac{e^{|g_k|}}{\alpha_k+e^{|g_k|}}\bigg] \\
&\ge \frac{1}{\alpha_k+1}\left(\frac{1}{2}\E e^{-|g_k|} + \frac{1}{2}\right)\\
&\ge \frac{1}{\alpha_k+1} \left(\frac{0.95}{2}e^{-2\varepsilon}+\frac{1}{2}\right) \quad \text{(as $\Prob{|g_k|\ge 2\varepsilon}\ge 0.95$ )}\\
&\ge \frac{0.67}{\alpha_k+1} \quad \text{(as $\varepsilon\le 1/2$)}.
\end{split}
\end{equation}
On the other hand, 
\begin{equation}
\label{eq:with_square}
\begin{split}
&\E\bigg[\frac{e^{2g_k}}{(\alpha_k+e^{g_k})^2}\bigg] = \frac{1}{2}\left(\E \bigg[\frac{e^{2|g_k|}}{(\alpha_k+e^{|g_k|})^2}\bigg]+ \E \bigg[\frac{e^{-2|g_k|}}{(\alpha_k+e^{-|g_k|})^2}\bigg] \right)\\
&\le \frac{1}{2}\left(2e^{2\varepsilon^2}\frac{1}{(\alpha_k+1)^2} + \E \bigg[\frac{e^{-2|g_k|}}{(\alpha_k+e^{-|g_k|})^2}\bigg] \right)\\
&= \frac{1}{2}\left(2e^{2\varepsilon^2}\frac{1}{(\alpha_k+1)^2} + \E \bigg[\frac{1}{(\alpha_ke^{|g_k|}+1)^2}\bigg] \right)\\
&= \left(e^{2\varepsilon^2}\frac{1}{(\alpha_k+1)^2} + \E \bigg[\frac{1}{2(\alpha_ke^{|g_k|}+1)^2}\bigg] \right)\\
&=\frac{1}{(\alpha_k+1)^2}\left(e^{2\varepsilon^2} + \frac{1}{2}\right)\\
&\le \frac{1}{(\alpha_k+1)^2}(\sqrt{e}+\frac{1}{2}) \quad \text{(as $\varepsilon\le 1/2$)}.
\end{split}
\end{equation}
Putting together \eqref{eq:with_square} and \eqref{eq:without_square} and the fact that $\varepsilon<1/2$, we obtain that 
\begin{equation*}
\left(\E\bigg[\frac{e^{2g_k}}{(\alpha_k+e^{g_k})^2}\bigg]\right)^{1/2}\le 2.18 \E\bigg[\frac{e^{g_k}}{\alpha_k+e^{g_k}}\bigg] = 2.18\E\bigg[\frac{e^{\ell(k)+g_k}}{e^{\ell(1)+g_1}+\ldots+e^{\ell(d)+g_d}}\bigg].
\end{equation*}
The proof is complete. 

\section*{Acknowledgments}
We would like to thank Felix Draxler and Victor Souza for helpful discussions. Also, we would like to thank Xuyang Chen for pointing out an inaccuracy in our first version. PA and RV are supported by NSF and Simons Research Collaborations on the Mathematical and Scientific Foundations of Deep Learning. RV is also supported by NSF Grant DMS 1954233.

\end{document}